\documentclass{article}

\usepackage{a4wide}
\usepackage{graphicx}
\usepackage{xspace}
\usepackage{amsmath}
\usepackage{amssymb}
\usepackage{microtype}

\graphicspath{{./images/}}

\newtheorem{theorem}{Theorem}
\newtheorem{corollary}[theorem]{Corollary}
\newtheorem{lemma}{Lemma}
\newenvironment{proof}{\emph{Proof.}}{\hfill $\Box$ \medskip\\}

\newcommand{\etal}{\emph{et al.}\xspace}

\newcommand{\graph}{half-\ensuremath{\theta_6}-graph\xspace}
\renewcommand{\c}[1]{\ensuremath{c(#1)}}
\newcommand{\ci}[2]{\ensuremath{c_{#1}(#2)}}
\newcommand{\canon}[2]{\ensuremath{T_{#1 #2}}}
\newcommand{\degreeNine}{\ensuremath{G_9}\xspace}
\newcommand{\degreeSix}{\ensuremath{G_6}\xspace}

\newcommand{\Vis}{\mathord{\it Vis}}
\newcommand{\length}[1]{\ensuremath{|#1|}}

\title{On Plane Constrained Bounded-Degree Spanners\thanks{An extended abstract of this paper appeared in the proceedings of the 10th Latin American Symposium on Theoretical Informatics (LATIN 2012)~\cite{BFRV2012ConstrainedBounded}.}~\thanks{This work is supported in part by the Natural Science and Engineering Research Council of Canada, Carleton University's President's 2010 Doctoral Fellowship, the Ontario Ministry of Research and Innovation, and the Danish Council for Independent Research, Natural Sciences, grant DFF-1323-00247, and JST ERATO Grant Number JPMJER1201, Japan.}}

\author{Prosenjit Bose\thanks{Carleton University, Ottawa, Canada. \texttt{jit@scs.carleton.ca, sander@cg.scs.carleton.ca}} \and 
	Rolf Fagerberg\thanks{University of Southern Denmark, Odense, Demark. \texttt{rolf@imada.sdu.dk}} \and 
	Andr\'e van Renssen\thanks{University of Sydney, Sydney, Australia. \texttt{andre.vanrenssen@sydney.edu.au}} \and 
	Sander Verdonschot\footnotemark[3]}

\date{}

\begin{document}

\maketitle

\begin{abstract}
Let $P$ be a finite set of points in the plane and $S$ a set of non-crossing line segments with endpoints in $P$. The visibility graph of $P$ with respect to $S$, denoted $\Vis(P,S)$, has vertex set $P$ and an edge for each pair of vertices $u,v$ in $P$ for which no line segment of $S$ properly intersects $uv$. We show that the constrained \graph (which is identical to the constrained Delaunay graph whose empty visible region is an equilateral triangle) is a plane 2-spanner of $\Vis(P,S)$. We then show how to construct a plane 6-spanner of $\Vis(P,S)$ with maximum degree $6+c$, where $c$ is the maximum number of segments of $S$ incident to a vertex.
\end{abstract}

\section{Introduction}
A geometric graph $G$ is a graph whose vertices are points in the plane and whose edges are line segments between pairs of vertices. A graph $G$ is called plane if no two edges intersect properly. Every edge is weighted by the Euclidean distance between its endpoints. The distance between two vertices $u$ and $v$ in $G$, denoted by $d_G(u, v)$ or simply $d(u, v)$ when $G$ is clear from the context, is defined as the sum of the weights of the edges along the shortest path between $u$ and $v$ in $G$. A subgraph $H$ of $G$ is a $t$-spanner of $G$ (for $t\geq 1$) if for each pair of vertices $u$ and $v$, $d_H(u, v) \leq t \cdot d_G(u, v)$. The smallest value $t$ for which $H$ is a $t$-spanner is the {\em spanning ratio} or {\em stretch factor} of $H$. The graph $G$ is referred to as the {\em underlying graph} of $H$. The spanning properties of various geometric graphs have been studied extensively in the literature (see \cite{BS11,NS-GSN-06} for a comprehensive overview of the topic). However, most of the research has focused on constructing spanners where the underlying graph is the complete Euclidean geometric graph. We study this problem in a more general setting with the introduction of line segment {\em constraints}.

Specifically, let $P$ be a set of vertices in the plane and let $S$ be a set of line segments with endpoints in $P$, with no two line segments intersecting properly. The line segments of $S$ are called \emph{constraints}. Two vertices $u$ and $v$ can \textit{see each other} if and only if either the line segment $uv$ does not properly intersect any constraint or $u v$ is itself a constraint. If two vertices $u$ and $v$ can see each other, the line segment $uv$ is a \emph{visibility edge}. The \emph{visibility graph} of $P$ with respect to a set of constraints $S$, denoted $\Vis(P,S)$, has $P$ as vertex set and all visibility edges as edge set. In other words, it is the complete graph on $P$ minus all edges that properly intersect one or more constraints in~$S$.

This setting has been studied extensively within the context of motion planning amid obstacles. Clarkson \cite{Cl87} was one of the first to study this problem and showed how to construct a linear-sized $(1+\epsilon)$-spanner of $\Vis(P,S)$. Subsequently, Das \cite{D97} showed how to construct a spanner of $\Vis(P,S)$ with constant spanning ratio and constant degree. Bose and Keil \cite{BK06} showed that the Constrained Delaunay Triangulation is a 2.42-spanner of $\Vis(P,S)$. In this article, we show that the constrained \graph (which is identical to the constrained Delaunay graph whose empty visible region is an equilateral triangle) is a plane 2-spanner of $\Vis(P,S)$ by generalizing the approach used by Bose~\etal~\cite{BFRV15}. This improves the upper bound on the spanning ratio of 36 implied by Bose~\etal~\cite{BCR2016GeneralizedDelaunay}. A key difficulty in proving this result stems from the fact that the constrained Delaunay graph is {\bf not} necessarily a triangulation (see Figure~\ref{fig:NoTriangulation}). We then generalize the elegant construction of Bonichon~\etal~\cite{BGHP10} to show how to construct a plane 6-spanner of $\Vis(P,S)$ with maximum degree $6+c$, where $c=\max\{c(v)|v\in P\}$ and $c(v)$ is the number of constraints incident to a vertex~$v$.

\begin{figure}[ht]
  \begin{center}
    \includegraphics{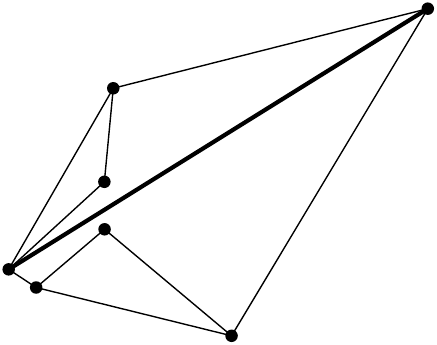}
  \end{center}
  \caption{The constrained \graph is not necessarily a triangulation. The thick line segment represents a constraint}
  \label{fig:NoTriangulation}
\end{figure}

\section{Preliminaries}
We define a \emph{cone} $C$ to be the region in the plane between two rays originating from a vertex referred to as the apex of the cone. We let six rays originate from each vertex, with angles to the positive $x$-axis being multiples of $\pi / 3$ (see Figure~\ref{fig:Cones}). Each pair of consecutive rays defines a cone. For ease of exposition, we only consider point sets in general position: no two vertices define a line parallel to one of the rays that define the cones and no three vertices are collinear. These assumptions imply that we can consider the cones to be open. If a point set is not in general position, one can easily find a suitable rotation of the point set to put it in general position. 

\begin{figure}[ht]
  \begin{minipage}[t]{0.47\linewidth}
    \begin{center}
      \includegraphics{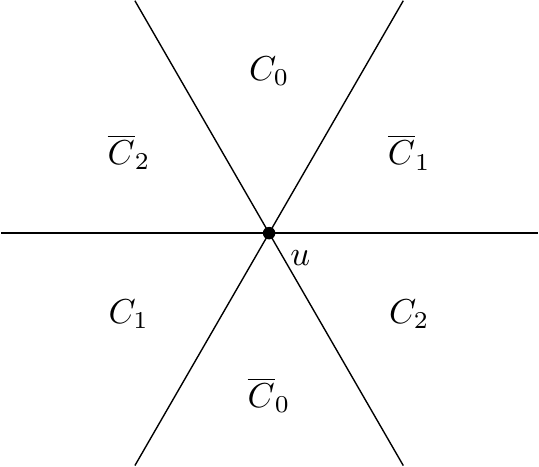}
    \end{center}
    \caption{The cones having apex $u$}
    \label{fig:Cones}
  \end{minipage}
  \hspace{0.05\linewidth}
  \begin{minipage}[t]{0.47\linewidth}
    \begin{center}
      \includegraphics{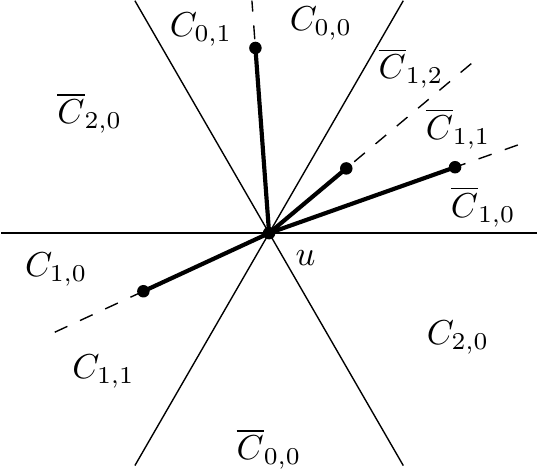}
    \end{center}
    \caption{The subcones having apex $u$. Constraints are shown as thick line segments}
    \label{fig:ConstrainedCones}
  \end{minipage}
\end{figure}

Let $(\overline{C}_1, C_0, \overline{C}_2, C_1, \overline{C}_0, C_2)$ be the sequence of cones in counterclockwise order starting from the positive $x$-axis. The cones $C_0$, $C_1$, and $C_2$ are called \emph{positive} cones and $\overline{C}_0$, $\overline{C}_1$, and $\overline{C}_2$ are called \emph{negative} cones. By using addition and subtraction modulo 3 on the indices, positive cone $C_i$ has negative cone $\overline{C}_{i+1}$ as clockwise next cone and negative cone $\overline{C}_{i-1}$ as counterclockwise next cone. A similar statement holds for negative cones. We use $C^u_i$ and $\overline{C}^u_j$ to denote cones $C_i$ and $\overline{C}_j$ with apex $u$. Note that for any two vertices $u$ and $v$, $v \in C^u_i$ if and only if $u \in \overline{C}^v_i$. 

Let vertex $u$ be an endpoint of a constraint $c$ and let the other endpoint $v$ lie in cone $C^u_i$. The lines through all such constraints $c$ split $C^u_i$ into several parts. We call these parts \emph{subcones} and denote the $j$-th subcone of $C^u_i$ by $C^u_{i, j}$, numbered in counterclockwise order (see Figure~\ref{fig:ConstrainedCones}). When a constraint $c = (u, v)$ splits a cone of $u$ into two subcones, we define $v$ to lie in both of these subcones. We call a subcone of a positive cone a positive subcone and a subcone of a negative cone a negative subcone. We consider a cone that is not split to be a single subcone.

We now introduce the constrained \graph, a generalized version of the \graph as described by Bonichon \etal\cite{BGHI10}: for each positive subcone of each vertex $u$, add an edge from $u$ to the closest vertex in or on the boundary of that subcone that can see $u$, where distance is measured along the bisector of the original cone (not the subcone) (see Figure~\ref{fig:Projection}). More formally, we add an edge between two vertices $u$ and $v$ if $v$ can see $u$, $v \in C^u_{i, j}$, and for all vertices $w \in C^u_{i, j}$ that can see $u$, $|u v'| \leq |u w'|$, where $v'$ and $w'$ denote the projection of $v$ and $w$ on the bisector of $C^u_i$ and $|xy|$ denotes the length of the line segment between two vertices $x$ and $y$. Note that our assumption of general position implies that each vertex adds at most one edge to the graph for each of its positive subcones. 

\begin{figure}[ht]
  \begin{minipage}[t]{0.47\linewidth}
    \begin{center}
      \includegraphics{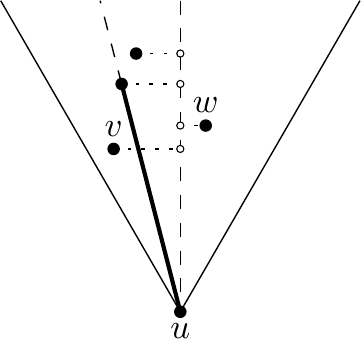}
    \end{center}
    \caption{Three vertices are projected onto the bisector of a cone of $u$. Vertex $v$ is the closest vertex in the left subcone and $w$ is the closest vertex in the right subcone}
    \label{fig:Projection}
  \end{minipage}
  \hspace{0.05\linewidth}
  \begin{minipage}[t]{0.47\linewidth}
    \begin{center}
      \includegraphics{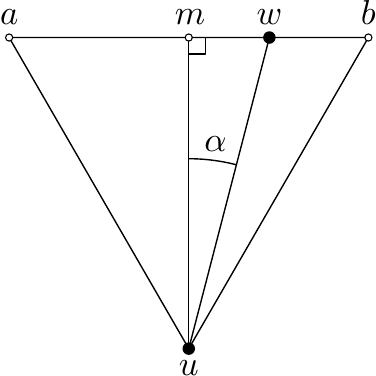}
    \end{center}
    \caption{Canonical triangle \canon{u}{w}}
    \label{fig:CanonicalTriangle1}
  \end{minipage}
\end{figure}

Given a vertex $w$ in a positive cone $C_i$ of vertex $u$, we define the \emph{canonical triangle} $\canon{u}{w}$ to be the triangle defined by the borders of $C^u_i$ and the line through $w$ perpendicular to the bisector of $C^u_i$ (see Figure~\ref{fig:CanonicalTriangle1}). Note that for each pair of vertices there exists a unique canonical triangle. We say that a region is \emph{empty} if it does not contain any vertices.

\section{Spanning Ratio of the Constrained Half-$\boldsymbol{\theta_6}$-Graph}
In this section we show that the constrained \graph is a plane 2-spanner of the visibility graph $\Vis(P,S)$. To do this, we first prove a property of visibility graphs. Recall that a region is \emph{empty} if it does not contain any vertices. 

\begin{lemma}
  \label{lem:ConvexChain}
  Let $u$, $v$, and $w$ be three arbitrary points in the plane such that $u w$ and $v w$ are visibility edges and $w$ is not the endpoint of a constraint intersecting the interior of triangle $u v w$. Then there exists a convex chain of visibility edges from $u$ to $v$ in triangle $u v w$, such that the polygon defined by $u w$, $w v$ and the convex chain is empty and does not contain any constraints.
\end{lemma}
\begin{proof}
  Let $Q$ be the set of vertices of $\Vis(P,S)$ inside triangle $u v w$. If $Q$ is empty, no constraint can cross $u v$, since one of its endpoints would have to be inside $u v w$, so our convex chain is simply $u v$. Otherwise, we build the convex hull of $Q \cup \{u, v\}$. Note that $u v$ is part of the convex hull since $Q$ lies inside $u v w$ to one side of the line through $u v$. When we remove this edge, we get a convex chain from $u$ to $v$ in triangle $u v w$. By the definition of a convex hull, the polygon defined by $u w$, $w v$ and the convex chain is empty.

  \begin{figure}[ht]
    \begin{center}
      \includegraphics{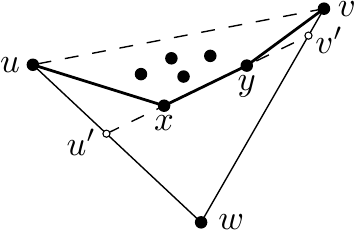}
    \end{center}
    \caption{A convex chain from $u$ to $v$ and intersections $u'$ and $v'$ of the triangle and the line through $x$ and $y$}
    \label{fig:VisiblePointInsideTriangle}
  \end{figure}

  Next, we show that two consecutive vertices $x$ and $y$ along the convex chain can see each other. Let $u'$ be the intersection of $u w$ and the line through $x$ and $y$ and let $v'$ be the intersection of $v w$ and the line through $x$ and $y$ (see Figure~\ref{fig:VisiblePointInsideTriangle}). Since $w$ is not the endpoint of a constraint intersecting the interior of triangle $u v w$ and, by construction, both $u'$ and $v'$ can see $w$, any constraint crossing $x y$ would need to have an endpoint inside $u' w v'$. But the polygon defined by $u w$, $w v$ and the convex chain is empty, so this is not possible. Therefore $x$ can see $y$.

  Finally, since the polygon defined by $u w$, $w v$ and the convex chain is empty and consists of visibility edges, any constraint intersecting its interior needs to have $w$ as an endpoint, which is not allowed. Hence, the polygon does not contain any constraints. 
\end{proof}

In the proof of Lemma~\ref{lem:ConvexChain}, note that $u$, $v$, and $w$ actually need not be part of the point set $P$. The lemma holds for any three points in the plane satisfying the requirements, if one considers the visibility edge as a line segment between any two points in the plane which is not intersected by a constraint. Lemma~\ref{lem:ConvexChain} will sometimes be used with this interpretation in mind later in the paper. Using this lemma, we proceed to improve the upper bound on the spanning ratio of the constrained \graph implied by Bose~\etal~\cite{BCR2016GeneralizedDelaunay}. 

\begin{theorem}
  \label{theo:ConstrainedSpanningRatio}
  Let $u$ and $w$ be vertices, with $w$ in a positive cone of $u$, such that $u w$ is a visibility edge. Let $m$ be the midpoint of the side of \canon{u}{w} opposing $u$, and let $\alpha$ be the unsigned angle between the lines $uw$ and $um$. There exists a path connecting $u$ and $w$ in the constrained \graph of length at most $(\sqrt{3} \cdot \cos \alpha + \sin \alpha) \cdot |u w|$ that lies inside \canon{u}{w}.
\end{theorem}
\begin{proof} 
  We assume without loss of generality that $w \in C^u_{0, j}$. We prove the theorem by induction on the area of \canon{u}{w}. Formally, we perform induction on the rank, when ordered by area, of the triangles \canon{x}{y} for all pairs of vertices $x$ and $y$ that can see each other. Let $\delta(x, y)$ denote the length of the shortest path from $x$ to $y$ in the constrained \graph that lies inside \canon{x}{y}. Let $a$ and $b$ be the upper left and right corner of \canon{u}{w}, and let $A$ and $B$ be the triangles $u a w$ and $u b w$ (see Figure~\ref{fig:SpanningTriangle}). Our inductive hypothesis is the following: 
  \begin{itemize}
    \item If $A$ is empty, then $\delta(u, w) \leq |ub| + |bw|$. 
    \item If $B$ is empty, then $\delta(u, w) \leq |ua| + |aw|$. 
    \item If neither $A$ nor $B$ is empty, then $\delta(u, w) \leq \max\{|ua| + |aw|, |ub| + |bw|\}$. 
  \end{itemize}

  We first note that this induction hypothesis implies the theorem: using the side of \canon{u}{w} as the unit of length, we have that $\delta(u, w) \leq (\sqrt{3} \cdot \cos\alpha + \sin\alpha) \cdot \length{uw}$ (see Figure~\ref{fig:CanonicalTriangle}).

  \begin{figure}[ht]
    \begin{minipage}[t]{0.31\linewidth}
      \begin{center}
	\includegraphics{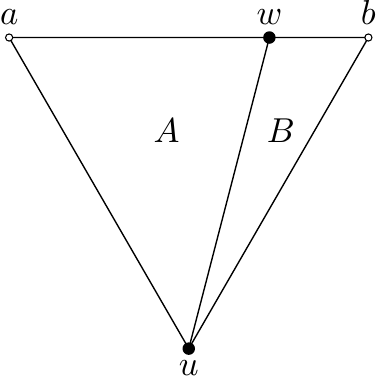}
      \end{center}
      \caption{Triangles $A$ and $B$}
      \label{fig:SpanningTriangle}
    \end{minipage}
    \hspace{0.01\linewidth}
    \begin{minipage}[t]{0.31\linewidth}
      \begin{center}
	\includegraphics{CanonicalTriangle}
      \end{center}
      \caption{Canonical triangle \canon{u}{w}}
      \label{fig:CanonicalTriangle}
    \end{minipage}
    \hspace{0.01\linewidth}
    \begin{minipage}[t]{0.31\linewidth}
      \begin{center}
	\includegraphics{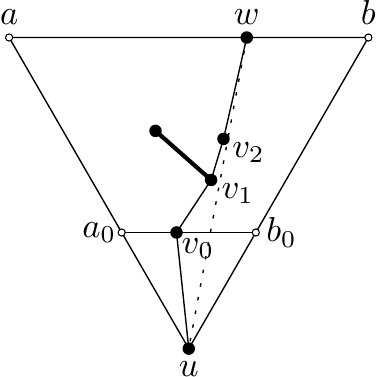}
      \end{center}
      \caption{Convex chain from $v_0$ to $w$}
      \label{fig:ConvexChain}
    \end{minipage}
  \end{figure}

  \textbf{Base case:} Triangle $\canon{u}{w}$ has minimal area. Since the triangle is a smallest canonical triangle, $w$ is the closest vertex to $u$ in its positive subcone. Hence the edge $u w$ is in the constrained \graph, and $\delta(u, w) = |u w|$. From the triangle inequality, we have that $|u w| \leq \min\{|u a| + |a w|, |u b| + |b w|\}$, so the induction hypothesis holds.

  \textbf{Induction step:} We assume that the induction hypothesis holds for all pairs of vertices that can see each other and have a canonical triangle whose area is smaller than the area of $\canon{u}{w}$. If $u w$ is an edge in the constrained \graph, the induction hypothesis follows by the same argument as in the base case. If there is no edge between $u$ and $w$, let $v_0$ be the visible vertex closest to $u$ in the positive subcone containing $w$, and let $a_0$ and $b_0$ be the upper left and right corner of $\canon{u}{v_0}$ (see Figure~\ref{fig:ConvexChain}). By definition, $\delta(u, w) \leq |u v_0| + \delta(v_0, w)$, and by the triangle inequality, $|u v_0| \leq \min\{|u a_0| + |a_0 v_0|, |u b_0| + |b_0 v_0|\}$. We assume without loss of generality that $v_0$ lies to the left of $u w$, which means that $A$ is not empty. 

  Since $u w$ and $u v_0$ are visibility edges, by applying Lemma~\ref{lem:ConvexChain} to triangle $v_0 u w$, a convex chain $v_0, ..., v_k = w$ of visibility edges connecting $v_0$ and $w$ exists (see Figure~\ref{fig:ConvexChain}). Note that, since $v_0$ is the closest visible vertex to $u$, every vertex along the convex chain lies above the horizontal line through $v_0$. 

  When looking at two consecutive vertices $v_{i-1}$ and $v_i$ along the convex chain, there are three types of configurations: (i) $v_{i-1} \in C^{v_i}_1$, (ii) $v_i \in C^{v_{i-1}}_0$ and $v_i$ lies to the right of or has the same $x$-coordinate as $v_{i-1}$, (iii) $v_i \in C^{v_{i-1}}_0$ and $v_i$ lies to the left of $v_{i-1}$. Let $A_i = v_{i-1} a_i v_i$ and $B_i = v_{i-1} b_i v_i$, the vertices $a_i$ and $b_i$ will be defined for each case. By convexity, the direction of $\overrightarrow{v_i v_{i+1}}$ is rotating counterclockwise for increasing $i$. Thus, these configurations occur in the order Type~(i), Type~(ii), and Type~(iii) along the convex chain from $v_0$ to $w$. We bound $\delta(v_{i-1}, v_i)$ as follows (see Figure~\ref{fig:ConvexChainCases}):

  \textbf{Type~(i):} If $v_{i-1} \in C^{v_i}_1$, let $a_i$ and $b_i$ be the upper left and lower corner of \canon{v_i}{v_{i-1}}. Triangle $B_i$ lies between the convex chain and $u w$, so it must be empty by Lemma~\ref{lem:ConvexChain}. Since $v_i$ can see $v_{i-1}$ and \canon{v_i}{v_{i-1}} has smaller area than \canon{u}{w}, the induction hypothesis gives that $\delta(v_{i-1}, v_i)$ is at most $|v_{i-1} a_i| + |a_i v_i|$.

  \begin{figure}[ht]
    \begin{center}
      \includegraphics{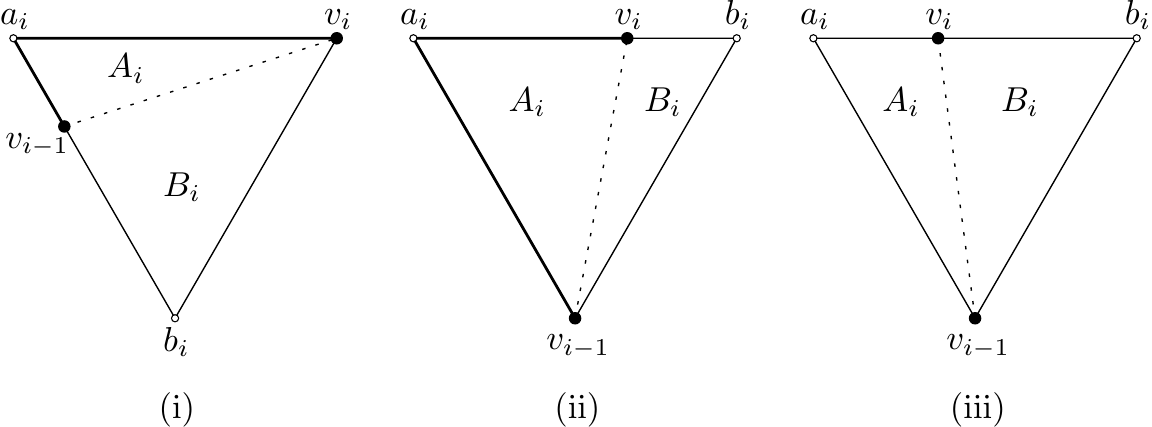}
    \end{center}
    \caption{Charging the three types of configurations}
    \label{fig:ConvexChainCases}
  \end{figure} 

  \textbf{Type~(ii):} If $v_i \in C^{v_{i-1}}_0$, let $a_i$ and $b_i$ be the left and right corner of \canon{v_{i-1}}{v_i}. Since $v_i$ can see $v_{i-1}$ and \canon{v_{i-1}}{v_i} has smaller area than \canon{u}{w}, the induction hypothesis applies. Whether $A_i$ and $B_i$ are empty or not, $\delta(v_{i-1}, v_i)$ is at most $\max\{|v_{i-1} a_i| + |a_i v_i|, |v_{i-1} b_i| + |b_i v_i|\}$. Since $v_i$ lies to the right of or has the same $x$-coordinate as $v_{i-1}$, we know $|v_{i-1} a_i| + |a_i v_i| \geq |v_{i-1} b_i| + |b_i v_i|$, so $\delta(v_{i-1}, v_i)$ is at most $|v_{i-1} a_i| + |a_i v_i|$.

  \textbf{Type~(iii):} If $v_i \in C^{v_{i-1}}_0$ and $v_i$ lies to the left of $v_{i-1}$, let $a_i$ and $b_i$ be the left and right corner of \canon{v_{i-1}}{v_i}. Since $v_i$ can see $v_{i-1}$ and \canon{v_{i-1}}{v_i} has smaller area than \canon{u}{w}, we can apply the induction hypothesis. Thus, if $B_i$ is empty, $\delta(v_{i-1}, v_i)$ is at most $|v_{i-1} a_i| + |a_i v_i|$ and if $B_i$ is not empty, $\delta(v_{i-1}, v_i)$ is at most $|v_{i-1} b_i| + |b_i v_i|$. 

  Recall that $a$ and $b$ are the upper left and right corner of \canon{u}{w} and that $B$ is the triangle $ubw$ (see Figure~\ref{fig:SpanningTriangle}). To complete the proof, we consider three cases: \mbox{(a) $\angle a w u \leq \pi/2$,} (b) $\angle a w u > \pi/2$ and $B$ is empty, (c) $\angle a w u > \pi/2$ and $B$ is not empty. 

  \textbf{Case (a):} If $\angle a w u \leq \pi/2$, the convex chain cannot contain any Type~(iii) configurations: for Type~(iii) configurations to occur, $v_i$ needs to lie to the left of $v_{i-1}$. However, by construction, $v_i$ lies to the right of the line through $v_{i-1}$ and $w$. Hence, since $\angle a w v_{i-1} < \angle a w u \leq \pi/2$, $v_i$ lies to the right of $v_{i-1}$. We can now bound $\delta(u, w)$ as follows using the bounds on Type~(i) and Type~(ii) configurations outlined above (see Figure~\ref{fig:SpanningProofCase1}):
  \begin{eqnarray*}
     \delta(u, w) &\leq& |u v_0| + \sum_{i=1}^k \delta(v_{i-1}, v_i) \\[-1ex]
		  &\leq& |u a_0| + |a_0 v_0| + \sum_{i=1}^k (|v_{i-1} a_i| + |a_i v_i|) \\
		  &=& |u a| + |a w|
  \end{eqnarray*}
  We see that the latter is equal to $|u a| + |a w|$ as required.

  \begin{figure}[ht]
    \begin{center}
      \includegraphics{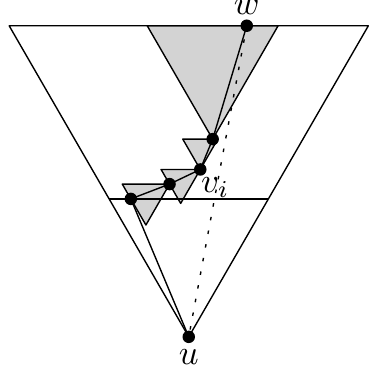}
      \includegraphics{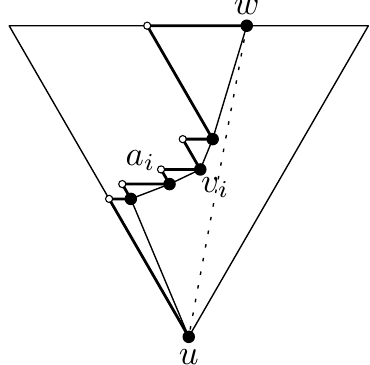}
      \includegraphics{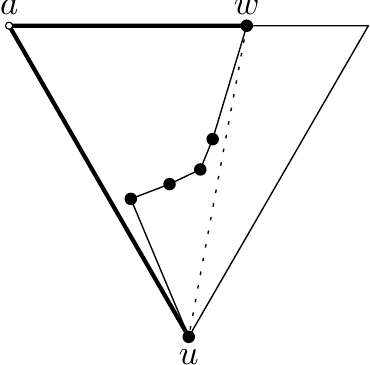}
    \end{center}
    \caption{Visualization of the paths (thick lines) in the inequalities of case (a)}
    \label{fig:SpanningProofCase1}
  \end{figure}

  \textbf{Case (b):} If $\angle a w u > \pi/2$ and $B$ is empty, the convex chain can contain Type~(iii) configurations. However, since $B$ is empty and the area between the convex chain and $u w$ is empty (by Lemma~\ref{lem:ConvexChain}), all triangles $B_i$ are also empty. Hence using the induction hypothesis, $\delta(v_{i-1}, v_i)$ is at most $|v_{i-1} a_i| + |a_i v_i|$ for all $i$. Using these bounds on the lengths of the paths between the vertices along the convex chain, we can bound $\delta(u, w)$ as in the previous case. Therefore, $\delta(u, w) \leq |u a| + |a w|$ as required. 

  \textbf{Case (c):} If $\angle a w u > \pi/2$ and $B$ is not empty, the convex chain can contain Type~(iii) configurations. Since $B$ is not empty, the triangles $B_i$ need not be empty. Recall that $v_0$ lies in $A$, hence neither $A$ nor $B$ are empty. Therefore, it suffices to prove that $\delta(u, w) \leq \max\{|ua| + |aw|, |ub| + |bw|\} = |ub| + |bw|$. Let \canon{v_j}{v_{j+1}} be the first Type~(iii) configuration along the convex chain (if it has any), let $a'$ and $b'$ be the upper left and right corner of \canon{u}{v_j}, and let $b''$ be the upper right corner of \canon{v_j}{w} (see Figure~\ref{fig:SpanningProofCase3}). Note that since $\angle a w u > \pi/2$ and $v_j$ lies to the left of $u w$, $|a' v_j|$ is smaller than $|b' v_j|$.
  \begin{eqnarray*}
    \delta(u, w) &\leq& |u v_0| + \sum_{i=1}^k \delta(v_{i-1}, v_i) \\[-1ex]
    &\leq& |u a_0| + |a_0 v_0| + \sum_{i=1}^j (|v_{i-1} a_i| + |a_i v_i|) + \sum_{i=j+1}^k (|v_{i-1} b_i| + |b_i v_i|) \\
    &=& |u a'| + |a' v_j| + |v_j b''| + |b'' w| \\
    &\leq& |u b'| + |b' v_j| + |v_j b''| + |b'' w| \\
    &=& |u b| + |b w| 
  \end{eqnarray*}

  \begin{figure}[ht]
    \begin{center}
      \includegraphics{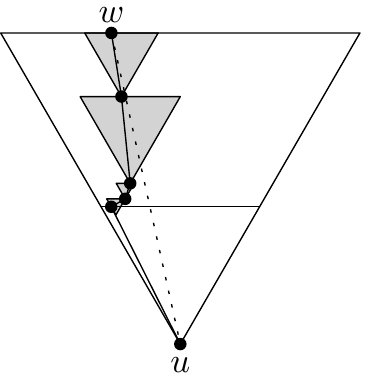}
      \includegraphics{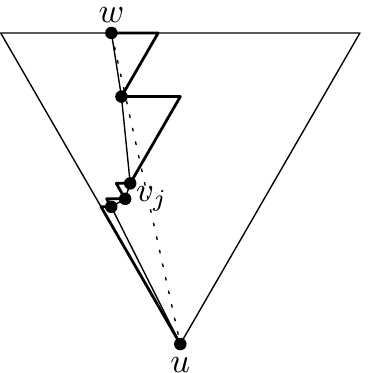}
      \includegraphics{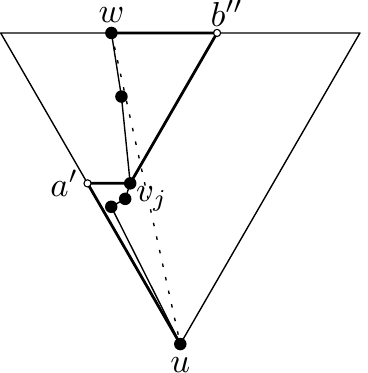}
      \includegraphics{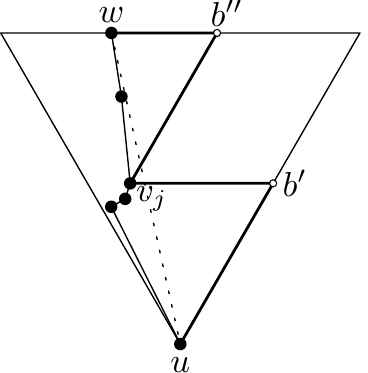}
      \includegraphics{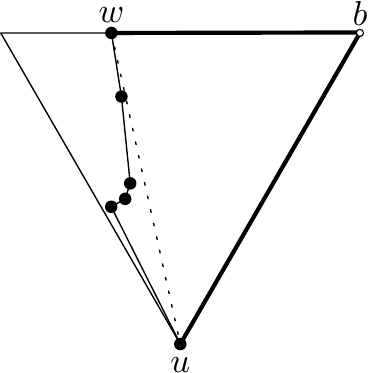}
    \end{center}
    \caption{Visualization of the paths (thick lines) in the inequalities of case (c)}
    \label{fig:SpanningProofCase3}
  \end{figure} \vspace{-2.3em} 
\end{proof}

Since the expression $\sqrt{3} \cdot \cos \alpha + \sin \alpha$ is increasing for $\alpha \in [0, \pi/6]$, the maximum value is attained by inserting the extreme value $\pi/6$. This leads to the following corollary. 

\begin{corollary}
  The constrained \graph is a 2-spanner of the visibility graph.
\end{corollary}

\noindent Next, we prove that the constrained \graph is plane. 

\begin{lemma}
  \label{lem:ContainsCorner}
  Let $u$, $v$, $x$, and $y$ be four distinct vertices such that the two canonical triangles \canon{u}{v} and \canon{x}{y} intersect. Then at least one of the corners of one canonical triangle is contained in the other canonical triangle. 
\end{lemma}
\begin{proof}
  If one triangle contains the other triangle, it contains all of its corners. Therefore we focus on the case where neither triangle contains the other. 

  By definition, the upper boundaries of \canon{u}{v} and \canon{x}{y} are parallel, the left boundaries of \canon{u}{v} and \canon{x}{y} are parallel, and the right boundaries of \canon{u}{v} and \canon{x}{y} are parallel. Because we assume that no two vertices define a line parallel to one of the rays that define the cones, we assume, without loss of generality, that the upper boundary of \canon{u}{v} lies below the upper boundary of \canon{x}{y}. The upper boundary of \canon{u}{v} must lie above the lower corner of \canon{x}{y}, since otherwise the triangles do not intersect. If the upper left (right) corner of \canon{u}{v} lies to the right (left) of the right (left) boundary of \canon{x}{y}, the triangles cannot intersect. Hence, either one of the upper corners of \canon{u}{v} is contained in \canon{x}{y} or the upper boundary of \canon{u}{v} intersects both the left and right boundary of \canon{x}{y}. In the latter case, the fact that the left boundaries of \canon{u}{v} and \canon{x}{y} are parallel and the right boundaries of \canon{u}{v} and \canon{x}{y} are parallel,  implies that the lower corner of \canon{x}{y} is contained in \canon{u}{v}. 
\end{proof}

\begin{lemma}
  \label{lem:Plane}
  The constrained \graph is plane.
\end{lemma}
\begin{proof}
  We prove the lemma by contradiction. Assume that two edges $u v$ and $x y$ cross at a point $p$. Since the two edges are contained in their canonical triangles, these triangles must intersect. By Lemma~\ref{lem:ContainsCorner} we know that at least one of the corners of one triangle lies inside the other. We focus on the case where the upper right corner of \canon{x}{y} lies inside \canon{u}{v}. The other cases are analogous. Since $u v$ and $x y$ cross, this also means that either $x$ or $y$ must lie in \canon{u}{v}. In the remainder, we assume that $y \in \canon{u}{v}$. The arguments used for the case where $x \in \canon{u}{v}$ are analogous. 

  \begin{figure}[ht]
    \begin{center}
      \includegraphics{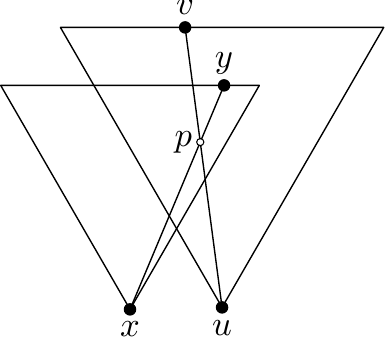}
    \end{center}
    \caption{Edges $u v$ and $x y$ intersect at point $p$}
    \label{fig:Planarity}
  \end{figure}  

  Assume without loss of generality that $v \in C^u_{0, j}$ (see Figure~\ref{fig:Planarity}). If $y \in C^u_{0, j}$, we look at triangle $u p y$. Since both $u$ and $y$ can see $p$, we get by Lemma~\ref{lem:ConvexChain} that either $u$ can see $y$ or $u p y$ contains a vertex. In both cases, $u$ can see a vertex in this subcone that is closer than $v$, contradicting the existence of the edge $u v$.

  If $y \notin C^u_{0, j}$, there exists a constraint $u z$ such that $v$ lies to one side of the line through $u z$ and $y$ lies on the other side. Since this constraint cannot cross $yp$, $z$ lies inside $u p y$ and is therefore closer to $u$ than $v$. Since by definition $z$ can see $u$, this also contradicts the existence of $u v$. 
\end{proof}

\section{Bounding the Maximum Degree}
In this section, we show how to construct a bounded degree subgraph \degreeNine of the constrained \graph that is a 6-spanner of the visibility graph. Given a vertex $u$ and one of its negative subcones, we define the \emph{canonical sequence} of this subcone as the vertices in this subcone that are neighbors of $u$ in the constrained \graph, in counterclockwise order (see Figure~\ref{fig:ConstructingDegree9}). These vertices all have $u$ as their closest visible vertex in a positive subcone. The \emph{canonical path} is defined by connecting consecutive vertices in the canonical sequence. This definition differs slightly from the one used by Bonichon \etal~\cite{BGHP10}.

\begin{figure}[ht]
  \begin{center}
    \includegraphics{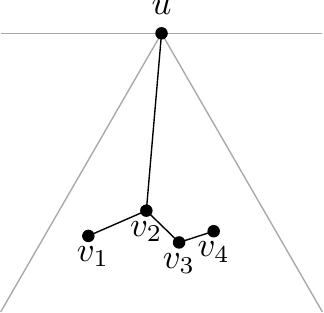}
  \end{center}
  \caption{The edges that are added to \degreeNine for a negative subcone of a vertex $u$ with canonical sequence $v_1, v_2, v_3$ and $v_4$}
  \label{fig:ConstructingDegree9}
\end{figure}

To construct \degreeNine, we start with a graph with vertex set $P$ and no edges. Then for each negative subcone of each vertex $u \in P$, we add the canonical path and an edge between $u$ and the closest vertex along this path, where distance is measured using the projections of the vertices onto the bisector of the cone containing the subcone. A given edge may be added by several vertices, but it appears only once in \degreeNine. This construction is similar to the construction of the unconstrained \mbox{degree-9} \graph described by Bonichon~\etal~\cite{BGHP10}. We proceed to prove that \degreeNine is a spanning subgraph of the constrained \graph with spanning ratio 3. 

\begin{lemma}\label{lem:Degree9IsSubgraph}
  \degreeNine is a subgraph of the constrained \graph.
\end{lemma}
\begin{proof}
  Given a vertex $u$, we look at one of its negative subcones, say $\overline{C}^u_{0, j}$. The edges added to \degreeNine for this subcone can be divided into two types: edges of the canonical path, and the edge between $u$ and the closest vertex along the canonical path. Since every vertex along the canonical path is by definition connected to $u$ in the constrained \graph, it remains to show that the edges of the canonical path are part of the constrained \graph.

  Let $v$ and $w$ be two consecutive vertices in the canonical path of $\overline{C}^u_{0, j}$, with $v$ before~$w$ in counterclockwise order. By applying Lemma~\ref{lem:ConvexChain} on the visibility edges $vu$ and $wu$, we get a convex chain $v = x_0, x_1, \dots , x_{k-1}, x_k = w$ of $k \geq 1$ visibility edges, which together with $vu$ and $wu$ form a polygon $Q$ empty of vertices and constraints. 

  Since $Q$ is empty, $v$ is not the endpoint of a constraint lying between $vu$ and $vx_1$. Hence, $x_1$ cannot be in cone $C_0^v$, otherwise $x_1$ would be closer to $v$ than $u$ in the subcone of $v$ that contains $u$. Similarly, $x_{k-1}$ cannot lie in cone $C_0^w$. By convexity of the chain, this implies that no vertex on the chain can lie in cone $C_0$ of another vertex on the chain. Hence, since $Q$ is empty, all vertices $x_i$ can see $u$. 

  We first show that $k = 1$, i.e. that the chain is just the line $vw$. We prove this by contradiction, so assume that $k > 1$. Hence, there is at least one vertex $x_i$ with $0 < i < k$. As such a vertex is not part of the canonical path in $\overline{C}^u_{0, j}$, it must see a closest vertex~$y$ different from $u$ in the subcone of $C^{x_i}_0$ that contains $u$. As vertices on the chain cannot lie in $C_0$ of each other, $y$ cannot be a vertex on the chain. As $Q$ is empty, $y$ must therefore lie strictly outside of~$Q$, and $yx_i$ must properly intersect either $vu$ or $wu$. But this contradicts the planarity of the constrained \graph, as $yx_i$, $vu$, and $wu$ would all be edges of this graph. Hence, $k=1$ and the chain is a single visibility edge~$vw$.

  It remains to show that $v w$ is an edge of the constrained \graph. Assume without loss of generality that $w$ lies in $C_{2}^v$ (the case that $v$ lies in $C_{1}^w$ is similar). We need to show that $w$ is the closest visible vertex in subcone $C_{2, j}^v$. We prove this by contradiction, so assume another vertex $x$ in $C_{2, j}^v$ is the closest. Vertex~$x$ lies in $\canon{v}{w}$, which is partitioned into a part inside~$Q$, a part to the right of~$wu$, and a part below~$vw$ (see Figure~\ref{fig:Degree9IsSubgraph}). If $x$ lies to the right of~$wu$, we would have intersecting edges $vx$ and $wu$, contradicting planarity of the constrained \graph. As $Q$ is empty, $x$ must lie below~$vw$ (see Figure~\ref{fig:Degree9IsSubgraph}).

  \begin{figure}[ht]
    \begin{center}
      \includegraphics{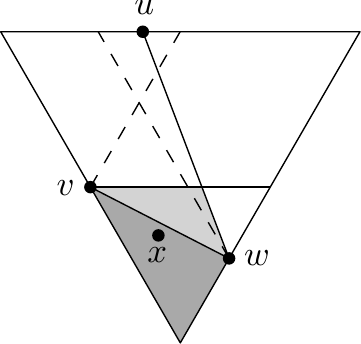}
    \end{center}
    \caption{$\canon{v}{w}$ is partitioned into a part inside~$Q$ (light gray), a part to the right of~$wu$ (white), and a part below~$vw$ (dark gray)}
    \label{fig:Degree9IsSubgraph}
  \end{figure}

  Applying Lemma~\ref{lem:ConvexChain} on the visibility edges $vx$ and $vw$, we get a convex chain $x = x_0, x_1, \dots , x_{k-1}, \\x_k = w$ of visibility edges and an empty polygon $R$. Vertex $x_1$ cannot lie in $C_{0}^x$, as this would contradict that $x$ is the closest visible vertex to $v$ in $C_{2, j}^v$. Hence, since $Q$ and $R$ are empty, $x$ can see $u$. Since $v$ and $w$ are two consecutive vertices in the canonical sequence of $\overline{C}^u_{0, j}$, $x$ is not part of this canonical sequence. So it must see a closest vertex~$y$ different from $u$ in the subcone of $C_{0}^x$ that contains $u$. Neither $v$ nor the convex chain from $x$ to $w$ lie in $C_{0}^x$. As $Q$ and $R$ are empty, $xy$ must properly intersect either $vu$ or $wu$, contradicting the planarity of the constrained \graph. 
\end{proof}

For future reference, we note that during the proof of Lemma~\ref{lem:Degree9IsSubgraph} the following two properties were shown.

\begin{corollary}
  \label{cor:NeighborsAlongCanonicalPath}
  Let $u$, $v$, and $w$ be three vertices such that $v$ and $w$ are neighbors along a canonical path of $u$ in $\overline{C}^u_{i}$. Vertex $w$ cannot lie in $C^v_{i}$ or $\overline{C}^v_{i}$. 
\end{corollary}

\begin{corollary}
  \label{cor:TriangleOfCanonicalPathEmpty}
  Let $u$, $v$, and $w$ be three vertices such that $v$ and $w$ are neighbors along a canonical path of $u$ in $\overline{C}^u_{i}$. Triangle $u v w$ is empty and does not contain any constraints. 
\end{corollary}

\begin{theorem}
  \label{theo:UnconstrainedDegree9SpanningRatio}
  \degreeNine is a 3-spanner of the constrained \graph.
\end{theorem}
\begin{proof} 
We prove the theorem by showing that for every edge $u w$ in the constrained \graph, where $w$ lies in a negative cone of $u$, \degreeNine contains a spanning path between $u$ and $w$ of length at most $3 \cdot |u w|$. This path will consist of a part of the canonical path in the subcone of $u$ that contains $w$ plus the edge between $u$ and the closest canonical vertex in that subcone.
      
  \begin{figure}[ht]
    \begin{center}
      \includegraphics{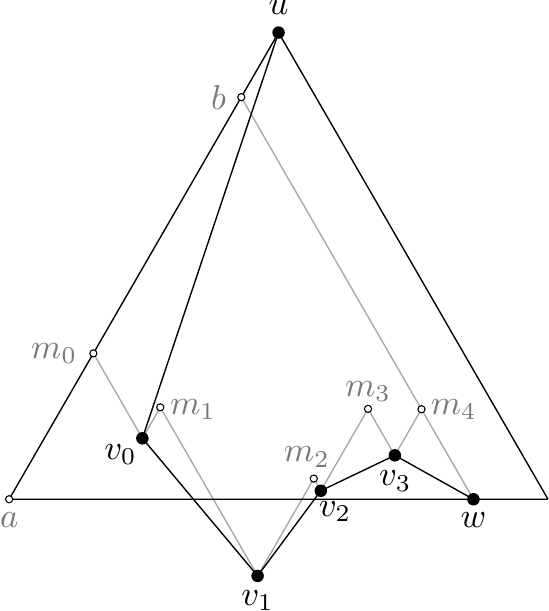}
    \end{center}
    \caption{Bounding the length of the canonical path}
    \label{fig:CanonicalPath}
  \end{figure}

  We assume without loss of generality that $w \in \overline{C}^u_0$. Let $v_0$ be the vertex closest to $u$ on the canonical path in the subcone $\overline{C}_{0, j}^u$ that contains $w$ and let $v_0, v_1, ..., v_k = w$ be the vertices along the canonical path from $v_0$ to $w$ (see Figure~\ref{fig:CanonicalPath}). Let $l_j$ and $r_j$ denote the rays defining the left and right boundaries of $C^{v_j}_0$ for $0 \leq j \leq k$ and let $r$ denote the ray defining the right boundary of $\overline{C}^u_0$ (as seen from $u$). Let $m_j$ be the intersection of $l_j$ and $r_{j-1}$, for $1 \leq j \leq k$, and let $m_0$ be the intersection of $l_0$ and $r$. Let $a$ be the intersection of $r$ and the horizontal line through $w$ and let $b$ be the intersection of $l_k$ and $r$. The length of the path between $u$ and $w$ in \degreeNine can now be bounded as follows:
  \begin{eqnarray*}
    d_{\degreeNine}(u, w) &\leq& |u v_0| + \sum_{j=1}^k |v_{j-1} v_j| \\ 
    &\leq& |u m_0| + |m_0 v_0| + \sum_{j=1}^k |m_j v_j| + \sum_{j=1}^{k} |v_{j-1} m_j| \\
    &=& |u m_0| + \sum_{j=0}^k |m_j v_j| + \sum_{j=1}^{k} |v_{j-1} m_j| 
  \end{eqnarray*}

  Since $u$ lies in $C_0$ of each of the vertices along the canonical path, all $m_j v_j$ project onto $w b$ and all $v_{j-1} m_j$ project onto $m_0 b$, when projecting along lines parallel to the boundaries of $\overline{C}^u_0$ instead of using orthogonal projections. By Corollary~\ref{cor:NeighborsAlongCanonicalPath} no edge on the canonical path can lie in $C_0$ of one of its endpoints, hence the projections of $m_j v_j$ onto $w b$ do not overlap. For the same reason, the projections of $v_{j-1} m_j$ onto $m_0 b$ do not overlap. Hence, we have that $\sum_{j=0}^k |m_j v_j| = |w b|$ and $\sum_{j=1}^{k} |v_{j-1} m_j| = |m_0 b|$. 
  \begin{eqnarray*}
    d_{\degreeNine}(u, w) &=& |u m_0| + \sum_{j=0}^k |m_j v_j| + \sum_{j=1}^{k} |v_{j-1} m_j| \\
    &=& |u m_0| + |w b| + |m_0 b| \\
    &\leq& |u a| + 2 \cdot |w a|
  \end{eqnarray*}

  Let $\alpha$ be $\angle a u w$. Using some basic trigonometry, we get $|u a| = |u w| \cdot \cos \alpha + |uw| \cdot \sin \alpha / \sqrt{3}$ and $|w a| = 2 \cdot |uw| \cdot \sin \alpha / \sqrt{3}$. Thus the spanning ratio can be expressed as:
  \begin{eqnarray*} 
    \frac{d_{\degreeNine}(u, w)}{|uw|} &\leq& \cos \alpha + 5 \cdot \frac{\sin \alpha}{\sqrt{3}} 
  \end{eqnarray*} 

  Since this is a non-decreasing function in $\alpha$ for $0 < \alpha \leq \pi/3$, its maximum value is obtained when $\alpha = \pi/3$, where the spanning ratio is 3. 
\end{proof} 

\noindent It follows from Theorems~\ref{theo:ConstrainedSpanningRatio} and \ref{theo:UnconstrainedDegree9SpanningRatio} that \degreeNine is a 6-spanner of the visibility graph. 

\begin{corollary}
  \degreeNine is a 6-spanner of the visibility graph.
\end{corollary}

To bound the degrees of the vertices, we use a charging scheme that charges the edges of a vertex to its cones. Summing the charge for all cones of a vertex then bounds its degree.

Recalling that the edges of \degreeNine are generated by canonical paths, consider a canonical path in $\overline{C}^u_{i, j}$, created by a vertex $u$. We use $v$ to indicate an arbitrary vertex along the canonical path, and we let $v'$ be the closest vertex to $u$ along the canonical path. The edges of \degreeNine generated by this canonical path are charged to cones as follows: 
\begin{itemize}
  \item The edge $u v'$ is charged to $\overline{C}^u_i$ and to $C^{v'}_i$. 
  \item An edge of the canonical path that lies in $\overline{C}^v_{i + 1}$ is charged to $C^v_i$. 
  \item An edge of the canonical path that lies in $\overline{C}^v_{i - 1}$ is charged to $C^v_i$. 
  \item An edge of the canonical path that lies in $C^v_{i + 1}$ is charged to $\overline{C}^v_{i - 1}$. 
  \item An edge of the canonical path that lies in $C^v_{i - 1}$ is charged to $\overline{C}^v_{i + 1}$. 
\end{itemize}
Essentially, the edge between $u$ and $v'$ is charged to the cones that contain it and edges along the canonical path are charged to the adjacent cone that is closer to the cone of $v$ that contains $u$. In other words, all charges are shifted one cone towards the positive cone containing $u$ (see Figure~\ref{fig:ChargingDegree9}). 

\begin{figure}[ht]
  \begin{center}
    \includegraphics{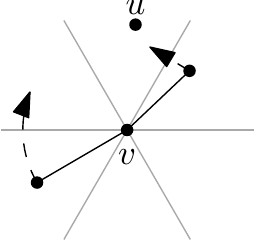}
  \end{center}
  \caption{Two edges of a canonical path and the associated charges}
  \label{fig:ChargingDegree9}
\end{figure}

By Corollary~\ref{cor:NeighborsAlongCanonicalPath}, no edge on the canonical path can lie in $C^v_{i}$ or $\overline{C}^v_{i}$, so the charging scheme above is exhaustive. Note that each edge is charged to both of its endpoints and therefore the charge on a vertex is an upper bound on its degree (only an upper bound, since an edge can be generated and charged by several canonical paths).

\begin{lemma}
\label{lem:BetweenConstraints}
  Let $v$ be a vertex that is incident to at least two constraints in the same positive cone $C^v_i$. Let $C^v_{i, j}$ be a subcone between two constraints and let $u$ be the closest visible vertex in this subcone. Let $\overline{C}^u_{i, k}$ be the subcone of $u$ that contains~$v$ and (when $uv$ is a constraint) intersects $C^v_{i, j}$. Then $v$ is the only vertex on the canonical path of $\overline{C}^u_{i, k}$.
\end{lemma}
\begin{proof}
  Let $v w_1$ and $v w_2$ be the two constraints between which subcone~$C^v_{i, j}$ lies. By applying Lemma~\ref{lem:ConvexChain} on these visibility edges, we get a convex chain $w_1 = x_0, x_1, \dots, x_k = w_2$ which together with $v w_1$ and $v w_2$ form a polygon $Q \subset C^v_{i, j}$ empty of vertices and constraints. Since $u$ is the closest vertex visible to~$v$ inside $C^v_{i, j}$, $u$ must be the vertex on this chain closest to~$v$. In particular, it is at least as close to~$v$ as $w_1$ and $w_2$. Since $v w_1$ and $v w_2$ are constraints and $Q$ is empty, there can be no vertex other than $v$ in $\overline{C}_{i, k}^u$ from which $u$ is visible. Hence, $v$ is the only vertex on the canonical path of $\overline{C}^u_{i, k}$. 
\end{proof}

\begin{lemma}
  \label{lem:ChargePositive}
  Each positive cone~$C_i$ of a vertex $v$ has a charge of at most $\max\{2, \ci{i}{v} + 1\}$, where \ci{i}{v} is the number of incident constraints in $C^v_i$.
\end{lemma}
\begin{proof}
  Let $u$ be a vertex such that $v$ is part of the canonical path of $u$. We first show that if this canonical path charges $C^v_i$, then $u$ must lie in $C^v_i$. Assume $u$ lies in $C^v_j$, $j \neq i$. Since all charges of this canonical path are shifted one cone towards $C^v_j$, a charge to $C^v_i$ would have to come from $\overline{C}^v_j$. However, by Corollary~\ref{cor:NeighborsAlongCanonicalPath}, no edge on the canonical path of a vertex in~$C^v_j$ can lie in $\overline{C}^v_j$.

  Next, we observe that there can be only one such vertex~$u$ for each subcone of $C^v_i$. This follows because $v$ is only part of canonical paths of vertices $u$ of which $u v$ is an edge in the constrained \graph, and there is at most one edge for each positive subcone.

  If $C^v_i$ is a single subcone and $v$ is not the closest vertex to $u$ on its canonical path, $C^v_i$ is charged for at most two edges along a single canonical path. Hence, its charge is at most 2. If $v$ is the closest vertex to $u$, the negative cones adjacent to this positive cone cannot contain any vertices of the canonical path. If they did, these vertices would be closer to $u$ than $v$ is, as distance is measured using the projection onto the bisector of the cone of $u$. Hence, if $v$ is the closest vertex to $u$, the positive cone containing $u$ is charged~1. Thus, when the positive cone is a single subcone, the cone is charged 2 if it has an edge of the canonical path in each adjacent negative cone, and at most 1 otherwise.

  Next, we look at the case where $C^v_i$ is not a single subcone. For each subcone, except the first and last, the canonical path of the vertex $u$ from that subcone consists only of $v$, by Lemma~\ref{lem:BetweenConstraints}. Hence, we get a charge of 1 per subcone and a charge of at most $\ci{i}{v} - 1$ in total for all subcones except the first and last subcone. We complete the proof by showing that the vertices~$u$ of the first and the last subcone can add a charge of at most 1 each.

  Consider the first subcone~$C^v_{i,0}$. The argument for the last subcone is symmetric. If $v$ is the closest vertex to $u$ on its canonical path, the negative cones adjacent to this positive cone cannot contain any vertices of the canonical path, since these would be closer to $u$ than $v$ is. Hence, the vertex~$u$ of this subcone adds a charge of 1. If $v$ is not the closest vertex to $u$, we argue that $v$ is the end of the canonical path of the vertex~$u$ of the subcone, implying that $u$ can add a charge of at most 1: Let $x$ be the other endpoint of the constraint that defines the subcone. Since $u$ is the closest visible vertex in this subcone of $v$, it cannot lie further from $v$ than $x$. If $u$ is $x$, constraint $u v$ splits $\overline{C}^u_i$ and only one of these two parts intersects the first subcone of $v$. Hence $v$ is the end of the canonical path of $u$. If $u$ is not $x$, $u$ lies closer to $v$ than $x$. Any vertex $y$ before~$v$ (in counterclockwise order) on the canonical path would have to lie in $C^v_{i+1}$ or $\overline{C}^v_{i-1}$, since by Corollary~\ref{cor:NeighborsAlongCanonicalPath}, $y$ cannot lie in $C^v_i$ or $\overline{C}^v_i$. Since $y$ must also lie in $\overline{C}^u_i$ to be on this canonical path, vertex~$u$ is not visible from $y$ due to the constraint~$x v$. Hence, no such vertex~$y$ can exist on the canonical path, implying that $v$ is the end of the canonical path.

  Summing up all charges, each positive cone is charged at most $\ci{i}{v} + 1$ if $\ci{i}{v} \geq 1$, and at most 2 otherwise. Hence, a positive cone is charged at most $\max\{2, \ci{i}{v} + 1\}$. 
\end{proof}

\begin{corollary}
  \label{cor:Charge2Constrained}
  If the $i$-th positive cone of a vertex $v$ has a charge of $\ci{i}{v} + 2$, then $\ci{i}{v} = 0$, i.e. it does not contain any constraints having $v$ as an endpoint in $C_i$ and is charged for two edges in the adjacent negative cones. 
\end{corollary}

\begin{lemma}
  \label{lem:ChargeNegative}
  Each negative cone~$\overline{C}_i$ of a vertex $v$ has a charge of at most $\ci{\overline{i}}{v} + 1$, where \ci{\overline{i}}{v} is the number of incident constraints in $\overline{C}^v_i$.
\end{lemma}
\begin{proof}
  A negative cone of a vertex $v$ is charged by the edge to the closest vertex in each of its subcones and it is charged by the two adjacent positive cones if edges of canonical paths lie in those cones (see Figure~\ref{fig:ConstrainedChargingDegree9NegativeConeProof}). We first show that vertices that do not lie in the positive subcones directly adjacent to $\overline{C}^v_i$ cannot have an edge involving $v$ along their canonical paths. Let $u$ be a vertex that does not lie in a positive subcone directly adjacent to $\overline{C}^v_i$ and let $v x$ be the constraint closest to $\overline{C}^v_i$ that defines the boundary of the subcone of $v$ that contains $u$. For $u$ to have an edge along its canonical path that is charged to $\overline{C}^v_i$, it needs to lie further from $u$ than $x$, since otherwise no vertex creating such an edge is visible to $u$. However, this implies that $v$ would not connect to $u$, thus it would not part of the canonical path of $u$. 

  \begin{figure}[ht]
    \begin{center}
      \includegraphics{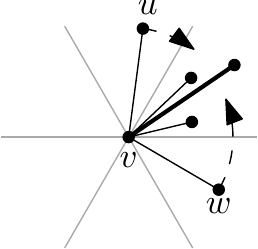}
    \end{center}
    \caption{If $v w$ is present, the negative cone does not contain edges having $v$ as endpoint}
    \label{fig:ConstrainedChargingDegree9NegativeConeProof}
  \end{figure}

  As $v$ can only be part of the canonical path of a single vertex in each of its positive subcones, we need to consider only the charges to $\overline{C}^v_i$ from the canonical path created by the closest visible vertices in the two positive subcones directly adjacent to $\overline{C}^v_i$. Let these vertices be u and w.

  Next, we show that every negative cone can be charged by at most one edge in total from its adjacent positive cones. Suppose that $w$ lies in a positive cone of $v$ and $v w$ is part of the canonical path of $u$. Then $w$ lies in a negative cone of $u$, which means that $u$ lies in a positive cone of $w$ and cannot be part of a canonical path for $w$. It remains to show that this negative cone of $v$ cannot be charged by an edge $v u'$ from a canonical path of a different vertex $w'$. Since $u v w$ forms a triangle in constrained \graph and this graph is planar, no edge of $u' v w'$ can cross any of the edges of $u v w$. This implies that either $u'$ and $w'$ lie inside $u v w$ or $u$ and $w$ lie inside $u' v w'$. However, by Corollary~\ref{cor:TriangleOfCanonicalPathEmpty}, triangles $x y z$ formed by a vertex $x$ and two vertices $y$ and $z$ that are neighbors along the canonical path of $x$ are empty. Therefore, $u'$ and $w'$ cannot lie inside $u v w$ and $u$ and $w$ cannot lie inside $u' v w'$. Thus every negative cone charged by at most one edge in total from its adjacent positive cones.

  Finally, we show that if one of $u v$ or $v w$ is present, the negative cone does not have an edge to the closest vertex in that cone and it contains no constraint that has $v$ as an endpoint. We first show that if one of $u v$ or $v w$ is present, the negative cone does not have an edge to the closest vertex in that cone. We assume without loss of generality that $v w$ is present, \mbox{$u \in C^v_i \cap C^w_i$}, and $w \in C^v_{i-1}$. Since $v$ and $w$ are neighbors on the canonical path of $u$, we know that the triangle $u v w$ is part of the constrained \graph and, by Corollary~\ref{cor:TriangleOfCanonicalPathEmpty}, this triangle is empty. Furthermore, since $u w$ is an edge of the constrained \graph and, by Lemma~\ref{lem:Plane}, the constrained \graph is plane, $v$ cannot have an edge to the closest vertex beyond~$u w$. Hence the negative cone does not have an edge to the closest vertex in that cone. By the same argument, the negative cone cannot contain a constraint that has $v$ as an endpoint. 

  It follows that if this negative cone contains no constraint that has $v$ as an endpoint, it is charged at most 1, by one of $u v$, $v w$, or the edge to the closest. Also, if this negative cone does contain constraints that have $v$ as an endpoint, it is not charged by edges in the adjacent positive cones and hence its charge is at most $\ci{\overline{i}}{v} + 1$, one for the closest in each of its subcones. 
\end{proof}

\begin{theorem}
  Every vertex $v$ in \degreeNine has degree at most $\c{v} + 9$. 
\end{theorem}
\begin{proof}
  From Lemmas~\ref{lem:ChargePositive} and \ref{lem:ChargeNegative}, each positive cone has charge at most $\ci{i}{v} + 2$ and each negative cone has charge at most $\ci{\overline{i}}{v} + 1$, where \ci{i}{v} and \ci{\overline{i}}{v} are the number of constraints in the $i$-th positive and negative cone. Since a vertex has three positive and three negative cones and the \ci{i}{v} and \ci{\overline{i}}{v} sum up to \c{v}, this implies that the total degree of a vertex is at most $\c{v} + 9$. 
\end{proof}

\section{Bounding the Maximum Degree Further}
In this section, we show how to reduce the maximum degree further, resulting in a plane 6-spanner~\degreeSix of the visibility graph in which the degree of any node~$v$ is bounded by $\c{v} + 6$.

By Lemmas~\ref{lem:ChargePositive} and \ref{lem:ChargeNegative} we see that if we can avoid the case where a positive cone gets a charge of $\ci{i}{v} + 2$, then every cone is charged at most $\ci{i}{v} + 1$, for a total charge of $\c{v} + 6$. By Corollary~\ref{cor:Charge2Constrained}, this case only happens when a positive cone has $\ci{i}{v} = 0$ and is charged for two edges in the adjacent negative cones. This situation is depicted in Figure~\ref{fig:PositiveConeCharge2}, where $x$, $v$, and $y$ are all on the canonical path of $u$. We construct \degreeSix by performing a transformation on \degreeNine for all positive cones in this situation.

\begin{figure}[ht]
  \begin{minipage}[t]{0.4\linewidth}
    \begin{center}
    \includegraphics{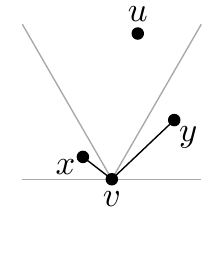}
  \end{center}
  \caption{A positive cone having charge 2}
  \label{fig:PositiveConeCharge2}
  \end{minipage}
  \hspace{0.05\linewidth}
  \begin{minipage}[t]{0.5\linewidth}
    \begin{center}
    \includegraphics{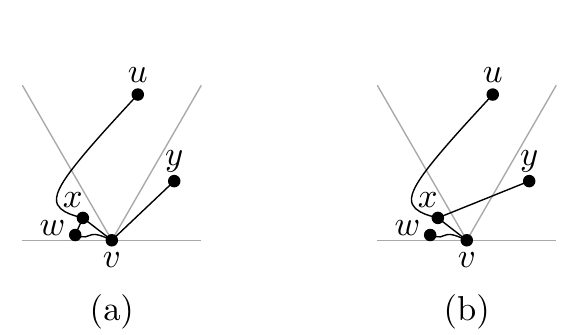}
  \end{center}
  \caption{Transforming \degreeNine (a) into \degreeSix (b)}
  \label{fig:ConstructingDegree6}
  \end{minipage}
\end{figure}

We now describe the transformation. We assume without loss of generality that the positive cone in question is $C^v_0$. We call a vertex $v$ the \emph{closest canonical vertex} in a negative subcone of $u$ when, among the vertices of the canonical path of $u$ in that subcone, $v$ is closest to $u$.

We first note that if $x$ is the closest canonical vertex in one of the at most two subcones of $\overline{C}^v_2$ 
that contain it, the edge $v x$ is charged to $C^v_0$, since $v x$ is an edge of the canonical path induced by $u$, and it is also charged to cone $\overline{C}^v_2$, since it is the closest canonical vertex in one of its subcones. Since we need to charge it only once to account for the degree of $v$, we can remove the charge to $C^v_0$, reducing its charge by 1 as desired. Similarly, if $y$ is the closest canonical vertex in one of the at most two subcones of $\overline{C}^v_1$
that contain it, it is charged to both $C^v_0$ and $\overline{C}^v_1$, so we can reduce the charge to $C^v_0$ by 1.
Therefore, we only perform a transformation if neither $x$ nor $y$ is the closest canonical vertex in the subcones of~$v$ that contain them. 

In that case, the transformation proceeds as follows. First, we add an edge between $x$ and $y$. Next, we look at the sequence of vertices between $v$ and the closest canonical vertex on the canonical path induced by $u$. If this sequence includes $x$, we remove $v y$. Otherwise we remove $v x$. Note that by Corollary~\ref{cor:TriangleOfCanonicalPathEmpty}, triangles $u x v$ and $u v y$ are empty and do not contain any constraints and therefore the edge $x y$ does not intersect any constraints. 

We assume without loss of generality that $v y$ is removed. By removing $v y$ and adding $x y$, we reduce the degree of $v$ at the cost of increasing the degree of $x$. Hence, we need to find a way to balance the degree of $x$. Since $x$ lies in $\overline{C}^v_2$ and the edge $x v$ is part of the constrained \graph, $x$ lies on a canonical path of $v$ in $\overline{C}^v_2$ and, since $x$ is not the closest canonical vertex to $v$ on this canonical path, $x$ has a neighbor $w$ along this canonical path. We claim that $x$ is the last vertex along the canonical path of $v$ in $\overline{C}^v_2$ and thus $w$ is uniquely defined. This follows because for any vertex $z$ later than~$x$ along that canonical path, either $z$ must lie in triangle $u v x$, contradicting its emptiness by Corollary~\ref{cor:TriangleOfCanonicalPathEmpty}, or the edges $zv$ and $xu$ of the constrained \graph must intersect, contradicting its planarity by Lemma~\ref{lem:Plane}. To balance the degree of $x$, we remove edge $x w$, if $w$ lies in $\overline{C}^x_0$ and $w$ is not the closest canonical vertex in a subcone of $\overline{C}^x_0$ 
that contains it. Otherwise $x w$ is not removed. The situation before the transformation is shown in Figure~\ref{fig:ConstructingDegree6}~(a) and the situation after the transformation is shown in Figure~\ref{fig:ConstructingDegree6}~(b). A curved line segment denotes a part of a canonical path plus the edge from its closest canonical vertex.

To construct \degreeSix, we apply this transformation on each positive cone matching the situation above. Note that since edge $uv$ is part of the constrained \graph, which is plane, and \degreeNine is a subgraph of the constrained \graph, the edges added by this transformation cannot be part of \degreeNine as they cross $uv$. Hence, since only edges of \degreeNine are removed, there are no conflicts among the transformations of different cones, i.e. no cone will add an edge that was removed by another cone and vice versa. Before we prove that this construction yields a graph of maximum degree $6+c$, we first show that the resulting graph is still a 3-spanner of the constrained \graph.

\begin{lemma}
  \label{lem:NotClosest}
Let $v x$ be an edge of \degreeNine and let $x$ lie in a negative cone $\overline{C}_i$ of $v$. If $x$ is not the closest canonical vertex in either of the at most two subcones of $\overline{C}^v_i$ that contain it, then the edge $v x$ is used by at most one canonical path. 
\end{lemma}
\begin{proof}
We prove the lemma by contraposition. Assume that edge $v x$ is part of two canonical paths of two vertices $u$ and $w$. For $v$ and $x$ to be neighbors on a canonical path of $u$ and $w$, these vertices need to lie in $C^v_{i+1} \cap C^x_{i+1}$ or $C^v_{i-1} \cap C^x_{i-1}$, by Corollary~\ref{cor:NeighborsAlongCanonicalPath}. By Corollary~\ref{cor:TriangleOfCanonicalPathEmpty} and planarity of the constrained \graph, $u$ and $w$ cannot lie in the same region, hence one lies in $C^v_{i+1} \cap C^x_{i+1}$ and one lies in $C^v_{i-1} \cap C^x_{i-1}$. We assume without loss of generality that $u \in C^v_{i+1} \cap C^x_{i+1}$ and $w \in C^v_{i-1} \cap C^x_{i-1}$ (see Figure~\ref{fig:NotClosest}). Thus $u v x$ and $w v x$ form two disjoint triangles in the constrained \graph and, by Corollary~\ref{cor:TriangleOfCanonicalPathEmpty}, both triangles are empty. Furthermore, since the constrained \graph is plane, no edge from $v$ can cross $u x$ or $w x$, making $v x$ the only edge of $v$ in $\overline{C}_i$. Therefore, $x$ is the closest canonical vertex in any subcone of $\overline{C}^v_i$ that contains it.

  \begin{figure}[ht]
    \begin{center}
    \includegraphics{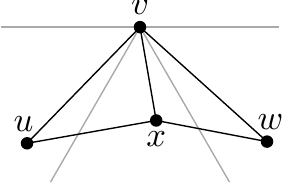}
  \end{center}
  \caption{If edge $v x$ is part of two canonical paths, $x$ is the only neighbor of $v$ in the negative cone of $v$}
  \label{fig:NotClosest}
  \end{figure}
\end{proof}

\begin{lemma}
   For every edge $u w$ in the constrained \graph, there exists a path in \degreeSix of length at most $3 \cdot |u w|$.
\end{lemma}
\begin{proof}
  In the proof of Theorem~\ref{theo:UnconstrainedDegree9SpanningRatio} we showed that for every edge $u w$ in the constrained \graph, where $w$ lies in a negative cone of $u$, \degreeNine contains a spanning path between $u$ and $w$ of length at most $3 \cdot |u w|$, consisting of a part of the canonical path in the subcone of $u$ that contains $w$ plus the edge between $u$ and the closest canonical vertex in that subcone. We now show that \degreeSix also contains a spanning path between $u$ and $w$ of length at most $3 \cdot |u w|$. 

Note that in the construction, we never remove an edge $v x$ with $x$
being the closest canonical vertex in a negative subcone of $v$. This means two things: 1) For any spanning path in \degreeNine, its last edge still exists in \degreeSix. 2) By Lemma~\ref{lem:NotClosest}, any removed edge is part of a single canonical path, so we need to argue only about this single canonical path and the spanning paths using it.

During the construction of \degreeSix, two types of edges are removed: Type~1, represented by $v y$ in Figure~\ref{fig:ConstructingDegree6}, and Type~2, represented by $x w$ in Figure~\ref{fig:ConstructingDegree6}. We first show that no edge removal of either of these types removes edge $v x$ in Figure~\ref{fig:ConstructingDegree6}. A Type~1 removal that has $v$ as the middle vertex in the configuration, as shown in Figure~\ref{fig:ConstructingDegree6}, is called \emph{centered at} $v$. A Type~1 removal of $v y$ affects the single canonical path containing $x v$ and $v y$ (see Figure~\ref{fig:ConstructingDegree6}). 
  We note that no Type~1 removal involving $v$ can be centered at $x$ or $y$, since $v$ lies in a positive cone of both $x$ and $y$ and a Type~1 removal requires both neighbors of the center vertex to lie in negative cones. 
  This implies that Type~1 removals are non-overlapping (i.e. their configurations do not share edges) and, in particular, it implies that edge $v x$ is not removed by this type of removal. 

  A Type~2 removal of $x w$ affects the canonical path that contains $w$ and $x$ (see Figure~\ref{fig:ConstructingDegree6}). As argued during the construction of \degreeSix, $x$ is the last vertex along a canonical path of $v$ and the edge~$x w$ is removed if $w$ lies in a negative cone of $x$ and $w$ is not a closest canonical vertex to $x$. We now show that edge $v x$ cannot be removed by a Type 2 removal: For it to be removed, we need that either $x$ lies in a negative cone of $v$ and $v$ is the last vertex along this canonical path, or $v$ lies in a negative cone of $x$ and $x$ is the last vertex along this canonical path. However, since $v$ is not the last vertex along the canonical path that contains $v$ and $x$ (it is followed by $y$) and $v$ does not lie in a negative cone of $x$, neither condition is satisfied. 

  Now that we know that for every Type~1 removal, edge $v x$ is still present in the final \degreeSix, we look at the spanning paths in \degreeSix. Every spanning path present in \degreeNine can be affected by several non-overlapping Type~1 removals, as well as by a Type~2 removal at either end. By applying the triangle inequality to Figure~\ref{fig:ConstructingDegree6}, it follows that $|x y| \leq |x v| + |v y|$. Combined with the fact that for every Type~1 removal, $v x$ is present in \degreeSix, it follows that there still exists a spanning path between $u$ and any vertex $w$ along its canonical path, except possibly the last vertex $x$ on either end, as the edge connecting $x$ to its neighbor along the canonical path could be removed by a Type~2 removal. However, we perform a Type 2 removal only when $u$ and $x$ are part of a Type~1 configuration centered at $u$ and $u x$ is the edge of this configuration that was not removed (see Figure~\ref{fig:ConstructingDegree6}, where $v$ acts as the node called $u$ in the Type~2 argument above). Furthermore, we showed that in this case $u x$ is still present in \degreeSix. Hence, there exists a spanning path of length at most $3 \cdot |u w|$ between $u$ and any vertex $w$ along its canonical path. 

Thus, we have proven that for every edge $u w$ in the constrained \graph, where $w$ lies in a negative cone of $u$, also \degreeSix contains a spanning path between $u$ and $w$ of length at most $3 \cdot |u w|$. 
\end{proof}

\begin{lemma}
  Every vertex $v$ in \degreeSix has degree at most $\c{v} + 6$. 
\end{lemma}
\begin{proof}
  To bound the degree, we look at the charges of the vertices. We prove that after the transformation each positive cone has charge at most $\ci{i}{v} + 1$ and each negative cone has charge at most $\ci{\overline{i}}{v} + 1$. This implies that the total degree of a vertex is at most $\c{v} + 6$. Since the charge of the negative cones is already at most $\ci{\overline{i}}{v} + 1$, we focus on positive cones having charge $\ci{i}{v} + 2$. By Corollary~\ref{cor:Charge2Constrained}, this means that these cones have charge 2 and $\ci{i}{v} = 0$. 

  Let $v$ be a vertex such that one of its positive cones $C^v_i$ has charge 2, let $u$ be the vertex whose canonical path charged 2 to $C^v_i$, and let $x \in \overline{C}^v_{i-1}$ and $y \in \overline{C}^v_{i+1}$ be the neighbors of $v$ on this canonical path (see Figure~\ref{fig:PositiveConeCharge2}). If $x$ or $y$ is the closest canonical vertex in a subcone of $\overline{C}^v_{i-1}$ or $\overline{C}^v_{i+1}$, this edge has been charged to both that negative cone and $C^v_i$. Hence we can remove the charge to $C^v_i$ while maintaining that the charge is an upper bound on the degree of $v$. 

  If neither $x$ nor $y$ is the closest canonical vertex in a subcone of $\overline{C}^v_{i-1}$ or $\overline{C}^v_{i+1}$, edge $x y$ is added. We assume without loss of generality that edge $v y$ is removed. Thus $v y$ need not be charged, decreasing the charge of $C^v_i$ to 1. Since $v y$ was charged to $\overline{C}^y_{i-1}$ and this charge is removed, we charge edge $x y$ to $\overline{C}^y_{i-1}$. Thus the charge of $y$ does not change. 

  It remains to show that we can charge $x y$ to $x$.
Recall that $x$ lies on the canonical path of $v$ in $\overline{C}^v_{i-1}$, is the last vertex on this canonical path, and has $w$ as neighbor on this canonical path (see Figure~\ref{fig:ConstructingDegree6}).
Since vertices $u v x$ and $v w x$ each form a triangle of neighboring vertices on a canonical path in the constrained \graph, by Corollary~\ref{cor:TriangleOfCanonicalPathEmpty} they are empty and do not contain any constraints. This implies that $x$ is not the endpoint of any constraint in $C^x_{i-1}$. 
Hence, since $x$ is the last vertex along the canonical path of $v$,
$C^x_{i-1}$ has charge at most 1 by Lemma~\ref{lem:ChargePositive} and Corollary~\ref{cor:Charge2Constrained}. Now, consider the neighbor~$w$ of $x$. Vertex $w$ can be in one of two cones with respect to $x$: $C^x_{i+1}$ and $\overline{C}^x_i$. If $w \in C^x_{i+1}$, $x w$ is charged to $\overline{C}^x_i$. Thus the charge of $C^x_{i-1}$ is 0 and we can charge $x y$ to it. 

  If $w \in \overline{C}^x_i$ and $w$ is the closest canonical vertex to $x$ in a subcone of $\overline{C}^x_i$, $x w$ has been charged to both $C^x_{i-1}$ and $\overline{C}^x_i$. We can remove that charge from $C^x_{i-1}$ and instead charge $x y$ to it, while keeping the charge of $C^x_{i-1}$ at 1. If $w \in \overline{C}^x_i$ and $w$ is not the closest canonical vertex to $x$ in a subcone of $\overline{C}^x_i$ 
  that contains it, $x w$ was removed during the transformation. Since this edge was charged to $C^x_{i-1}$, we can now charge $x y$ to $C^x_{i-1}$, while keeping the charge of $C^x_{i-1}$ at 1. 
\end{proof}

\begin{lemma}
  \degreeSix is a plane subgraph of the visibility graph.
\end{lemma}
\begin{proof}
Since \degreeNine is a plane subgraph of the visibility graph by Lemmas~\ref{lem:Plane} and \ref{lem:Degree9IsSubgraph}, only the edges added in the transformation from \degreeNine to \degreeSix can violate the lemma. An added edge~$x y$ can potentially intersect edges of \degreeSix that are in the constrained \graph, other edges that were added (recall that added edges are not in the constrained \graph, so these two cases are disjoint), and constraints.

First, we consider intersections of $x y$ with edges of \degreeSix that are in the constrained \graph. Since $x y$ was added in the transformation, $x$, $y$, and $v$ are part of a canonical path of some vertex~$u$ (see Figure~\ref{fig:ConstructingDegree6}). Thus, in the constrained \graph $u v x$ and $u v y$ form two triangles, each containing a pair of neighboring vertices along the canonical path, which are empty by Corollary~\ref{cor:TriangleOfCanonicalPathEmpty}. Since the constrained \graph is plane and $x y$ lies inside $u x v y$, the only edge of the constrained \graph that can intersect $x y$ is $u v$. We now argue that $u v$ is not in~\degreeSix. By construction, $u v$ can only be part of \degreeNine if $v$ is the closest vertex to~$u$ on this canonical path, or if $u v$ are neighboring vertices along another canonical path of some vertex~$t$. The former cannot be the case, by the conditions for adding $x y$ in the transformation (see Figure~\ref{fig:ConstructingDegree6}). Assume the latter is the case. If $u \in C^v_i$, then either $t \in C^u_{i+1} \cap C^v_{i+1}$ or $t \in C^u_{i-1} \cap C^v_{i-1}$, by Corollary~\ref{cor:NeighborsAlongCanonicalPath}. If $t \in C^u_{i-1} \cap C^v_{i-1}$, the triangle $uvt$ contains all of $\overline{C}^u_{i} \cap \overline{C}^v_{i+1}$, which contains $y$, as shown in Figure~\ref{fig:DegreeSixPlaneProof}.

As $uvt$ is empty by Corollary~\ref{cor:TriangleOfCanonicalPathEmpty}, this is a contradiction. If $t \in C^u_{i+1} \cap C^v_{i+1}$, a similar contradiction based on~$x$ arises. This shows that $u v$ is not in~\degreeNine, and hence not in \degreeSix either, as edges added in the transformation from \degreeNine to \degreeSix are not in the constrained \graph.

  \begin{figure}[ht]
    \begin{center}
    \includegraphics{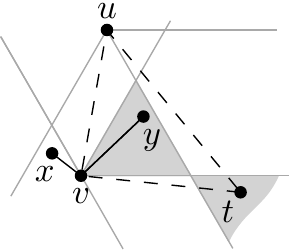}
  \end{center}
  \caption{If $t \in C^u_{i-1} \cap C^v_{i-1}$, the triangle $uvt$ contains all of $\overline{C}^u_{i} \cap \overline{C}^v_{i+1}$, which contains $y$}
  \label{fig:DegreeSixPlaneProof}
  \end{figure} 

Next, we consider intersections of $x y$ with other added edges. By Corollary~\ref{cor:TriangleOfCanonicalPathEmpty} the quadrilateral $u x v y$ does not contain any vertices. Its sides $u x$, $x v$, $v y$, and $y u$ are edges of the constrained \graph, which we showed above cannot be intersected by added edges. Hence, the only possibility for an added edge to intersect $x y$ is the edge $u v$. However, $u v$ cannot be an added edge, as we argued above. Thus, $x y$ cannot intersect an added edge.

  Finally, we consider intersection of $x y$ with constraints. By Corollary~\ref{cor:TriangleOfCanonicalPathEmpty}, triangles $u x v$ and $u v y$ are empty and do not contain any constraints. Hence, since edge $x y$ is contained in $u x v y$, it does not intersect any constraints. 
\end{proof}

\section{Conclusion}
We showed that the constrained \graph is a plane 2-spanner of $\Vis(P,S)$. We then generalized the construction of Bonichon~\etal~\cite{BGHP10} to show how to construct a plane 6-spanner of $\Vis(P,S)$ with maximum degree $6+c$, where $c=\max\{c(v)|v\in P\}$ and $c(v)$ is the number of constraints incident to a vertex $v$. 

A number of open problems still remain. For example, since constrained $\theta$-graphs with at least 6 cones were recently shown to be spanners~\cite{BR14}, a logical next question is to see if the method shown in this paper can be generalized to work for any constrained $\theta$-graph. It would also be interesting to see if the degree can be reduced further still, while remaining a spanner of $\Vis(P,S)$. 

Furthermore, it would be interesting to see if it is possible to reduce the maximum degree of the vertices further. This was recently shown to be possible in the unconstrained setting~\cite{bonichon2015there,kanj2017degree}, which raises the question whether the approaches used in the unconstrained setting work in the constrained setting as well. Since these two approaches use different graphs as a starting point and thus require different edge removal rules and shortcutting techniques, it could very well be the case that only one of them results in a plane graph that respects the constraints.

\bibliographystyle{plain} 
\bibliography{references}

\end{document}